\documentclass[12pt]{article}
\usepackage{enumitem}
\usepackage{amsfonts}
\usepackage{amsmath}
\usepackage{cases}
\usepackage{amssymb}
\usepackage{graphicx}
\usepackage{cite}
\usepackage{textcomp}%
\usepackage{algorithmic}
\usepackage{algorithm}
\usepackage{color}
\usepackage{amsthm}
\usepackage{epsfig}
\usepackage{epstopdf}
\usepackage{multirow}
\usepackage{caption}
\usepackage{subfigure}
\usepackage{latexsym}
\usepackage{arydshln}
\usepackage{indentfirst}
\usepackage{xcolor}
\usepackage{lscape}
\usepackage{geometry}
\numberwithin{equation}{section}

\newcommand{\RNum}[1]{\uppercase\expandafter{\romannumeral #1\relax}}

\newtheorem{theorem}{Theorem}[section]

\newtheorem{example}{Example}[section]
\newtheorem{open problem}{Open Problem}[section]

\newtheorem{lemma}{Lemma}[section]

\newtheorem{corollary}{Corollary}[section]

\allowdisplaybreaks

\begin{document}

\begin{center}{\Large \bf {The weight distributions of some linear
codes derived from Kloosterman sums}}\footnote {This work was supported by the National Natural Science Foundation of China
 under Grant 12371326. }\\
\vspace{0.5cm}
{Mengzhen Zhao \footnote{E-mail address: 22110491@bjtu.edu.cn}, Yanxun Chang \footnote{ E-mail address: yxchang@bjtu.edu.cn}}\\\

 {\small { {\it School of Mathematics and Statistics, Beijing Jiaotong University, Beijing 100044, China.}}}
 \vspace{0.8cm}
\end{center}

\begin{center}
\begin{minipage}{153mm}
{\bf \small  Abstract:} {\small{Linear codes with few weights have applications in data storage systems, secret
 sharing schemes, and authentication codes. In this paper, some kinds of $p$-ary linear
 codes with few weights are constructed  by use of the given defining set, where $p$ is a prime. Their weight distributions are determined  based on Kloosterman sums over finite fields. In addition, some linear codes we given is minimal.}}\\
{ \small \textbf{Keywords}: linear code, weight distribution, Kloosterman sum, minimal.\\}
\end{minipage}
\end{center}
\section{Introduction}
Let $q = p^m$ and $F_q$ denote the finite field with $q$ elements, where $p$ is a prime and $m$ is a positive integer. An $[n, k, d]$ linear code $C$ over $F_p$ is a $k$-dimensional subspace of $F^n_p$ with minimum (Hamming) distance $d$. Let $A_i$ denote the number of codewords with Hamming weight $i$ in a code $C$ of length $n$. The weight enumerator of $C$ is defined by $1 + A_1z + A_2z^2 +\cdots+ A_nz^n$. Accordingly, the sequence $(1, A_1, \ldots , A_n)$ is called the weight distribution of $C$. A code $C$ is said to be a $t$-weight code if the number of nonzero $A_i$ in $(1, A_1, \ldots , A_n)$ is equal to $t$. Linear codes with few weights are important in many aspects, like as secret sharing schemes\cite{1984-C} and association schemes \cite{2005-C}.
As an important class of linear codes, minimal linear codes have applications in secret sharing schemes \cite{1993-M}. Therefore, the constructions of linear codes with few weights
and minimal linear codes have attracted much attention in coding theory.

In \cite{2015-D}, Ding et al. proposed a new construction of linear codes by use of defining sets and obtained a class of two-weight or three-weight linear codes. Let $D_0=\{d_{1},d_{2},\ldots,d_{n}\}\subseteq F_{q}^*$. A linear code from defining set $D_0$ is defined by $$C_{D_0}=\{c(a)=(Tr_{1}^{m}(ax))_{x\in D_0}:a\in F_{q}\},$$ where $Tr^m_1(\cdot)$ is the absolute trace function from $F_q$ to $F_p$.
Later, many linear codes with few weights  have been constructed from different $D_0$ and functions like trace function \cite{2016-H,2018-T}, bent function \cite{2016-Z} and plateaued function \cite{2020-M} for example.

In \cite{2017-L}, Li et al. extended Ding's construction and defined a $p$-ary linear code by
\begin{equation}\label{form}
C_D = \{c(a, b) = (Tr_1^m(ax + by))_{(x,y)\in D} : a, b \in F_q\}
\end{equation}
where $D = \{(x_1, y_1),(x_2, y_2),\ldots,(x_n, y_n)\}\subseteq (F_q\times F_q)\setminus\{(0, 0)\}$. They employed exponential sums to investigate the weight distribution of the linear code $C_D$, where $D = \{(x, y) \in (F_q\times F_q)\setminus\{(0, 0)\}: Tr(x^{N_1} + y^{N_2}) = 0\}$ and $N_1, N_2 \in\{1, 2, p^{\frac{m}{2}+1}\}$. In \cite{2019-J}, Jian et al. extended the results in \cite{2017-L} and considered the linear codes $C_D$ where $D = \{(x, y) \in (F_q\times F_q)\setminus\{(0, 0)\}: Tr(x^{N_1} + y^{p^u+1}) = 0\}$, where $N_1\in\{1,2\}$, $p$ is odd  prime and $u$ is a positive integer satisfied some conditions. In this paper, we obtain the complete weight distribution of a class of linear code ${C}_{D}$ of the form (\ref{form}) with a given defining set $D$ by using Kloosterman sums. And we obtain two classes of linear codes which are minimal constructed by deleting codewords from the original linear code ${C}_{D}$.

The rest of this paper is organized as follows. In Section 2, we review some basic notations and basic results  needed in the subsequent sections. In Section 3, we  mainly present the complete weight distributions of three kinds of linear codes. In Section 4, we give the explicit proofs of the main  results in Section 3.  In Section 5, we conclude this paper.
\section{Preliminaries}
%In this section, we will give proofs of Theorem \ref{D} and Corollaries \ref{1D}, \ref{2D}. To complete these proofs, we need the following useful content.
Let $q=p^m$ and $m=2t$, where $p$ is a prime, $m,t$ are positive  integers and $t>2$.
Throughout this paper, we adopt the following notations unless otherwise stated:

\begin{itemize}
\item $Tr_{l_1}^{l_2}(x)$ is the trace function from the finite field $F_{p^{l_2}}$ to $F_{p^{l_1}}$ for any positive integer $l_1$ and $l_2$, where $l_1$ divides $l_2$.
\item $\xi_p=e^{2\pi\sqrt{-1}/p}$ is primitive $p$-th root of complex unity.
\item $\chi$ is the canonical additive character on $F_q$ , i.e., $\chi(x) = \xi_p^{Tr_1^m(x)}$ for any $x \in F_q$.
\item For any positive integer $l$, $\chi_{l}$ is the canonical additive character on $F_{p^l}$.

\end{itemize}

Let $\alpha$ be a generator of $F^*_{q}$, $\beta=\alpha^{p^{t}-1},$ and then  $order_q(\beta)=p^{t}+1$, where $order_q(\cdot)$ represents the order of any element in the multiplication group $F^*_q$.
Let $\Delta=\left<\beta\right>=\{\alpha^{j(p^{t}-1)}:0\leq j \leq p^{t}\}$, $\Gamma=\{\alpha^j:0\leq j\leq p^{t}\}$, and then $\Delta=\{v^{p^{t}-1}:v\in\Gamma\}.$
It is easy to know that $order_q(\alpha^{p^{t}+1})=p^{t}-1.$ Let $F^*_{p^{t}}=\left<\alpha^{p^{t}+1}\right>$ be a subgroup of $F_q^*$  and then $\Gamma$ is the set of the coset representative elements of $F_{p^{t}}^*$ of $F_q^*$. Thus, for each $x\in F^*_q$ , it has a unique decomposition as $x = uv,$ where $u \in F^*_{p^{t}}$ and $v \in \Gamma$.

With the above notations, we will give some useful results for later use.
The weights of the linear codes presented in this paper will be determined by the Kloosterman sums over finite field. For any $a\in F_{p^l}$, the Kloosterman sum of canonical additive character $\chi_{l}$
over $F_{p^l}$ at the point $a$ is defined by $$K_l(a)=\sum\limits_{x\in F_{p^l}^*}\chi_{l}(ax+\frac{1}{x}),$$ where $l$ is a positive integer.
The following results about Kloosterman sum of canonical additive character  will be useful in the sequel.

\begin{lemma}\label{bound}{\rm\cite[Theorem 5.45]{1997-L}} For any positive integer $l$ and  $a\in F^*_{p^l},$
$|K_l(a)|\leq2\sqrt{p^l}$.
\end{lemma}

Combining the proof of \cite[Theorem 3]{2006-N} about Kloosterman sum with  \cite[Theorem 1]{2006-H}, we summarize the result as follows with our notations in this paper and give a brief proof.

\begin{lemma}\label{equl}%\cite[Theorem.1]{2006-H}
For any $a\in F_{p^t}^*\subseteq F_{q}^*,$ we have $$\sum\limits_{x\in \Delta}\chi(ax)=-K_{t}(a^2).$$
\end{lemma}
\begin{proof}

By the proof in  \cite[Theorem 3]{2006-N}, for $p=2$, we have
\begin{eqnarray*}
% \nonumber % Remove numbering (before each equation)
  \sum\limits_{x\in \Delta}\chi(ax) &=& -1-\sum\limits_{\gamma\in a^{-1}F_{2^t}\setminus a^{-1}F_2}\chi_{t}(a^2\gamma+\frac{1}{\gamma})-\chi_{t}(a^2a^{-1}+\frac{1}{a^{-1}})+1\\
   &=& -K_t(a^2).
\end{eqnarray*}
By Theorem 1 in \cite{2006-H} and $a^{p^t-1}=1$ as $a\in F_{p^t}^*,$  for $p>2$, we have
\begin{eqnarray*}
% \nonumber to remove numbering (before each equation)
  \sum\limits_{x\in \Delta}\chi(ax) &=& \sum\limits_{j=0}^{p^t}\chi(a\alpha^{j(p^t-1)})
=\sum\limits_{j=0}^{p^t}\chi_{t}(Tr_{t}^{m}(a\alpha^{j(p^t-1)})) \\
   &=&-\sum\limits_{x\in F_{p^t}^*}\chi_t(x+ \frac{a^{p^t+1}}{x})=-\sum\limits_{x\in F_{p^t}^*}\chi_t(a^2x+ \frac{a^{p^t-1}}{x})\\
   &=&-\sum\limits_{x\in F_{p^t}^*}\chi_t(a^2x+ \frac{1}{x})=-K_{t}(a^2).
\end{eqnarray*}

\end{proof}

%\textcolor[rgb]{1.00,0.00,0.00}{Case 1: $p=2$. \\
%Page 740: $E=F_{2^k}\subset L=F_{2^{2k}}$, $K(a):=\sum_{x\in E}\chi_{E}(ax+\frac{1}{x})=1+\sum_{x\in E^*}\chi_{E}(ax+\frac{1}{x}).$
%$S=\{u\in L:u^{2^k+1}=1\}=F_{2^k}^*, S(\alpha)=\sum\limits_{u\in S}\chi_{L}(\alpha u).$ In Page 741, we have $S(a)=1-K(a)=-\sum_{x\in E^*}\chi_{E}(ax+\frac{1}{x})$. For $a\in F_p^*$, $K(a)=K(a^2)$.}
%
%Case 2: $p>2$.
%
%\textcolor[rgb]{1.00,0.00,0.00}{$\sum\limits_{x\in \Delta}\chi(ax)=\sum\limits_{v\in \Gamma}\chi(av^{r-1})=\sum\limits_{j=0}^{r}\chi(a\alpha^{j(r-1)})
%=\sum\limits_{j=0}^{r}Tr_{1}^{m}(a\alpha^{j(r-1)})=\sum\limits_{j=0}^{r}Tr_{1}^{k}(Tr_{k}^{m}(a\alpha^{j(r-1)}))=
%\sum\limits_{j=0}^{r}\chi(Tr_{k}^{m}(a\alpha^{j(r-1)}))=-K(\chi_{k};1,a^{r+1})=-K(\chi_{k};a^2,a^{r-1})=-K(\chi_{k};a^2,1)=-K_{k}(a^2).$
%}%$\sum\limits_{x\in F_{r}^*}\chi_{k}(a^2x+ \frac{1}{x})=\sum\limits_{x\in F_{r}^*}\chi_k(x+ \frac{a^{r+1}}{x}).$

\begin{lemma}{\rm\cite[Theorem 5.46]{1997-L}}
For any $a\in F_q^*,$ $$K_{t}(a)=-\sum\limits_{i=0}^{\lfloor\frac{t}{2}\rfloor}(-1)^{t-i}\frac{t}{t-i}\binom{t-i}{i}p^i(K_1(a))^{t-2i}.$$
\end{lemma}

In addition, the known Lemmas \ref{tanpan}, \ref{Pless} and \ref{character} are important in the proofs of our main results.

\begin{lemma}\label{tanpan}{\rm\cite[Lemma 1]{2018-T}}
For any given $a \in F^*_q,$ there is one and only one $v_a \in \Gamma$ such that $Tr_t^m(av_a) = 0$.
\end{lemma}

\begin{lemma}\label{Pless}{\rm \cite[p.259]{2003-H}}
For an $[n, k, d]$ code $C$ over $F_p$ with weight distribution $(1, A_1,\ldots, A_n),$ suppose that the weight distribution of its dual
code is $(1, A^{\bot}_1,\ldots, A^{\bot}_n),$ then the first two Pless power moments are
$$\sum\limits_{j=0}^{n}A_j=p^k,$$and
$$\sum\limits_{j=0}^njA_j=p^{k-1}(pn-n-A^{\bot}_1).$$
\end{lemma}

\begin{lemma}\label{character}{\rm\cite{1997-L}}
The orthogonal property of nontrivial additive character $\psi$ over finite field $E$ is given by
$$\sum\limits_{x\in E}\psi(zx)=\left\{\begin{array}{ll}
              0,& z\not=0,\\
             |E| ,&z=0.
             \end{array}\right.$$
\end{lemma}

\section{The weight distributions of some linear codes}
In this section, we only present the complete weight distributions of three kinds of linear codes. The explicit proofs of the following  results will be reveal in next section.
 Define a set
\begin{equation}\label{D}
  D=\{(x,y)\in (F_q\times F_q)\setminus\{(0,0)\}:Tr_1^m(x+y^{p^t-1})=0\},
\end{equation}
where $Tr_{1}^{m}(\cdot)$ is the trace function from the finite field $F_{q}$ to $F_{p}$. Then we can obtain some classes of linear codes by using $D$ as defining set.

\begin{theorem}
Let $C_D$ be the linear code over $F_p$ with defining set $D$ in $(\ref{D})$. Then $C_D$ is a four-weight
linear code with parameters $[n, 2m]$ and weight distribution in {\rm Table~\ref{tab1},} where $n=p^{2m-1}-1,$ and $$S=\sum\limits_{ z\in F_p^*}\sum\limits_{x\in \Delta}\chi(zx)=\sum\limits_{i=0}^{\lfloor t/2\rfloor}(-1)^{t-i}\frac{t}{t-i}\binom{t-i}{i}p^i\sum\limits_{z\in F_p^*}(K_1(z^2))^{t-2i}.$$

\begin{table}[ht]
\caption{Weight distribution of $C_D$}\label{tab1}%
\center{
\begin{tabular}{@{}ll @{}ll @{}}
\hline
Type&Weight & Frequency \\
\hline
$w_0$&0&1\\
$w_1$&$p^{2m-2}(p-1)~$~&$p^m(p^m-p+1)-1$\\
$w_2$&$(p^{2m-2}-p^{m-2})(p-1)-p^{m-2}(p^t-1)S~$&$p-1$\\
$w_3$&$(p^{2m-2}-p^{m-2})(p-1)-p^{m-2}(p^{t}(p-1)-S)~$~&$\frac{(p-1)(p^{m}-1)+(p^{t+1}-p^{t}-p+1)S}{p}$\\
$w_4$&$(p^{2m-2}-p^{m-2})(p-1)+p^{m-2}(p^{t}+S)~$~&$\frac{(p^2-2p+1)(p^m-1)-(p^{t+1}-p^{t}-p+1)S}{p}$\\
\hline
\end{tabular}}
\end{table}
\end{theorem}

\begin{example}
When $p=2,$ we know that $K_1(1)=1,$ so that $$S=\sum\limits_{i=0}^{\lfloor t/2\rfloor}(-1)^{t-i}\frac{t}{t-i}\binom{t-i}{i}2^i.$$
Let $t=3,$ then we have $S=5$. By magma, we can obtain the parameter of code $C_D$ is $[2047, 12, 448],$ and the weight enumerator is $1+x^{448}+49x^{960}+4031x^{1024}+14x^{1216},$ which agree with the distribution in {\rm Table \ref{tab1}}.
\end{example}

It is straightforward that the weight frequency of weight $w_2$ is equal to $|F_p^*|$. That inspires us to obtain two subcode ${C}_{D1}$ and $C_{D2}$ of $C_D$ by deleting some codewords, whose weights can be derived from $C_{D}$.

\begin{corollary}\label{1D}
Let $$C_{D1}=\{c(a,b)=(Tr_1^m(ax+by))_{(x,y)\in D}:a\in T,b\in F_q\},$$
where $T\subseteq F_{q}$ is a $F_p$-module such that $T\cap F_p=\{0\}$ and $|T|=p^r, 1\leq r \leq m-1,$
then $C_{D1}$ is a $[p^{2m-1}-1,m+r]$ linear code with constant weight $p^{2m-2}(p-1)$.
\end{corollary}

\begin{example}
Let $\alpha$ be the root of irreducible polynomial $f(x)=x^6+x+1$ over $F_2,$ $t=3$ and $T=\{\sum\limits_{i=1}^{5}a_i\alpha^i:a_i\in F_2\}$,
then $C_{D1}$ is a $[2047,11]$ linear code with constant weight $1024.$
\end{example}
%$$T=[0,v^1,v^2,v^3,v^4,v^5,v^7,v^8,v^9,v^{10},v^{13},v^{14},v^{15},v^{17},v^{19},v^{20},v^{25},v^{27},v^{28},$$$$v^{29},v^{33},v^{34},v^{36},v^{37},v^{39},v^{42},v^{46},v^{49},v^{50},v^{53},v^{55},v^{57}];
%,$$

\begin{corollary}\label{2D}
Let $\alpha$ be a root of  monic irreducible polynomial $f(x)$ over $F_p$ of degree $m,$ $c_{i,1}=(Tr_1^m(\alpha^ix))_{(x,y)\in D},$ $c_{i,2}=(Tr_1^m(\alpha^iy))_{(x,y)\in D}$ and $c_0=(Tr_1^m(x+y))_{(x,y)\in D}$.  Let $C_{D2}$ be a linear code generated by $\{c_{i,1},c_{i,2}:1\leq i\leq m-1\}\cup\{c_0\}$. Then $C_{D2}$ is a three-weight
linear code with parameters $[p^{2m-1}-1, 2m-1]$ and weight distribution in {\rm Table \ref{tab2}}.
%\begin{center}
%\begin{table}[h]
%\centering
%\caption{Weight distribution of $C_{D2}$}\label{tab2}%
%\begin{tabular}{@{}ll@{}ll@{}}
%\hline
%Type&Weight & Frequency \\
%\hline
%$w_0$&0&1\\
%$w_1$&$p^{2m-2}(p-1)~$&$p^{2m-1}-p^{m-1}-1$\\
%$w_3$&$(p^{2m-2}-p^{m-2})(p-1)-p^{m-2}(p^t(p-1)-S)~$&$p^{m-2}+p^{t-2}(S-p+1)$\\
%$w_4$&$(p^{2m-2}-p^{m-2})(p-1)+p^{m-2}(p^t+S)~$&$p^{m-1}-p^{m-2}-p^{t-2}(S-p+1)$\\
%\hline
%\end{tabular}
%\end{table}
%\end{center}

\begin{center}
\begin{table}[h]
\centering
\caption{Weight distribution of $C_{D2}$}\label{tab2}%
\begin{tabular}{@{}ll@{}ll@{}}
\hline
Type&Weight & Frequency \\
\hline
$w_0$&0&1\\
$w_1$&$p^{2m-2}(p-1)~$&$p^{2m-1}-1-p^{m-1}(p-1)$\\
$w_3$&$(p^{2m-2}-p^{m-2})(p-1)-p^{m-2}(p^t(p-1)-S)~$&$(p-1)(p^t+S-p+1)p^{t-2}$\\
$w_4$&$(p^{2m-2}-p^{m-2})(p-1)+p^{m-2}(p^t+S)~$&$(p-1)(p^{t+1}-p^{t}+p-1-S)p^{t-2}$\\
\hline
\end{tabular}
\end{table}
\end{center}
\end{corollary}

\begin{example}
Let $\alpha$ be a root of irreducible polynomial $f(x)=x^6+x+1$ over $F_2$ and $t=3$. By magma, we can obtain the parameter of code $C_{D2}$ is $[2047, 11, 960],$ and the weight enumerator is $1 +24x^{960} +2015x^{1024} +8x^{1216},$ which
agree with the distribution in {\rm Table \ref{tab2}}.
\end{example}

%\textcolor[rgb]{1.00,0.00,0.00}{Answer:}
%$|S|\leq\sum_{y\in F_p^*}|K_{k}(y^2)|\leq2(p-1)\sqrt{p^k}.$\\
%\begin{eqnarray*}
%% \nonumber % Remove numbering (before each equation)
%w_1&=&p^{2m-2}(p-1),\\
%w_3&=&(p^{2m-2}-p^{m-2}-p^{m+k-2})(p-1)+p^{m-2}S;\\
%w_4&=&(p^{2m-2}-p^{m-2})(p-1)+p^{m+k-2}+p^{m-2}S.
%\end{eqnarray*}
%
%\begin{eqnarray*}
%% \nonumber % Remove numbering (before each equation)
%w_1-w_3&=&\\
%&=&p^{2m-2}(p-1)-((p^{2m-2}-p^{m-2}-p^{m+k-2})(p-1)\\
%&&+p^{m-2}S)\\
%&\geq&p^{2m-2}(p-1)-((p^{2m-2}-p^{m-2}-p^{m+k-2})(p-1)\\
%&&+2p^{m-2}(p-1)\sqrt{p^k});\\
%&=&(p^{m-2}+p^{m+k-2})(p-1)-2p^{m-2}(p-1)\sqrt{p^k}\\
%&>&0.
%\end{eqnarray*}
%
%\begin{eqnarray*}
%% \nonumber % Remove numbering (before each equation)
%w_4-w_1&=&\\
%&=&(p^{2m-2}-p^{m-2})(p-1)+p^{m+k-2}+p^{m-2}S\\
%&&-p^{2m-2}(p-1)\\
%&\geq&(p^{2m-2}-p^{m-2})(p-1)+p^{m+k-2}-2p^{m-2}(p-1)\sqrt{p^k}\\
%&&-p^{2m-2}(p-1)\\
%&>&0.
%\end{eqnarray*}
%
%\begin{eqnarray*}
%% \nonumber % Remove numbering (before each equation)
%\frac{w_{min}}{w_{max}}=\frac{w_3}{w_4}&=&
%\frac{(p^{2m-2}-p^{m-2}-p^{m+k-2})(p-1)+p^{m-2}S}{(p^{2m-2}-p^{m-2})(p-1)+p^{m+k-2}+p^{m-2}S}\\
%&>&\frac{p-1}{p}.
%\end{eqnarray*}

\section{Proofs of main results}
In this section, we will give proofs of Theorem \ref{D}, Corollaries \ref{1D} and \ref{2D}. To complete these proofs, we first give some results which are necessary for calculating the lengths and weights of linear codes given in the Section 3.

\begin{lemma}\label{length}
Let $$n=|\{(x,y)\in (F_{q}\times F_q)\setminus \{(0,0)\}:Tr_{1}^{m}(x+y^{p^t-1})=0\}|,$$ then $n=p^{2m-1}-1.$
\end{lemma}

\begin{proof}
$n$ can be expressed in terms of character sums as
\begin{eqnarray*}
% \nonumber to remove numbering (before each equation)
n&=&\sum\limits_{x,y\in F_{q}}\frac{1}{p}\sum\limits_{z\in F_{p}}\xi_p^{zTr_{1}^{m}(x+y^{p^t-1})}-1\\
&=&\frac{q^2}{p}-1+\frac{1}{p}\sum\limits_{z\in F^*_{p}}\sum\limits_{x\in F_{q}}\chi(zx)\sum\limits_{y\in F_{q}}\chi(zy^{p^t-1})\\
%&=&\frac{q^2}{p}-1\\
&{=}&p^{2m-1}-1.~\ { \rm (by~Lemma~\ref{character})}
\end{eqnarray*}
\end{proof}

\begin{lemma}\label{weight}
Let $$N(a,b)=|\{(x,y)\in (F_{q}\times F_q)\setminus \{(0,0)\}:Tr_1^m(x+y^{p^t-1})=0,Tr_1^m(ax+by)=0\}|,$$ then
$$N(a,b)=\left\{\begin{array}{ll}
p^{2m-2}-1,& a\not\in F_p^*,\\
p^{2m-2}-1+p^{m-2}(p-1)+p^{m-2}(p^t-1)S,&a\in F_p^*,b=0,\\
p^{2m-2}-1+(p^{m-2}+p^{m+t-2})(p-1)-p^{m-2}S ,&a\in F_p^*, Tr_1^m(v_{b}^{p^t-1})=0,\\
p^{2m-2}-1+p^{m-2}(p-1)-p^{m+t-2}-p^{m-2}S ,&a\in F_p^*,Tr_1^m(v_{b}^{p^t-1})\not=0,
\end{array}\right.$$
where $v_b$ is the only one element in $\Gamma$ such that $Tr^m_t (bv_b)=0$ due to {\rm Lemma \ref{tanpan}}.
\end{lemma}

\begin{proof}
$N(a,b)$ can be expressed in terms of character sums as
\begin{eqnarray*}
% \nonumber % Remove numbering (before each equation)
N(a,b)
&=&\frac{1}{p^2}\sum\limits_{x,y\in F_{q}}\sum\limits_{z_{1}\in F_{p}}\xi_p^{z_{1}Tr_1^m(x+y^{p^t-1})}\sum\limits_{z_{2}\in F_{p}}\xi_p^{z_{2}Tr_1^m(ax+by)}-1\\
&=&\frac{1}{p^2}\sum\limits_{x,y\in F_q}(1+\sum\limits_{z_{1}\in F^*_{p}}\xi_p^{z_{1}Tr_1^m(x+y^{p^t-1})})(1+\sum\limits_{z_{2}\in F^*_{p}}\xi_p^{z_{2}Tr_1^m(ax+by)})-1\\
&=&\frac{q^2}{p^2}-1+\frac{1}{p^2}\sum\limits_{x,y\in F_{q}}\sum\limits_{z_{1}\in F^*_{p}}\chi(z_1x+z_1y^{p^t-1})\sum\limits_{z_{2}\in F^*_{p}}\chi(z_2ax+z_2by)\\
&+&\frac{1}{p^2}\sum_{x,y\in F_{q}}\sum\limits_{z_{1}\in F^*_{p}}\chi(z_1x+z_1y^{p^t-1})+\frac{1}{p^2}\sum\limits_{x,y\in F_{q}}\sum\limits_{z_{2}\in F^*_{p}}\chi(z_2ax+z_2by)\\
&=&\frac{q^2}{p^2}-1+\frac{1}{p^2}\sum\limits_{x,y\in F_{q}}\sum\limits_{z_{1}\in F^*_{p}}\chi(z_1x+z_1y^{p^t-1})\sum\limits_{z_{2}\in F^*_{p}}\chi(z_2ax+z_2by)~\ { \rm (by~Lemma~\ref{character})}\\
&=&\frac{q^2}{p^2}-1+\frac{1}{p^2}\sum\limits_{z_1,z_2\in F^*_{p}}\sum\limits_{x\in F_{q}}\chi((z_1+z_2a)x)\sum\limits_{y\in F_{q}}\chi(z_1y^{p^t-1}+z_2by)\\
&=&\frac{q^2}{p^2}-1+\frac{1}{p^2}\Omega.
\end{eqnarray*}

When $z_1+z_2a\not=0$, we have $\sum\limits_{x\in F_{q}}\chi((z_1+z_2a)x)=0$ for $z_1,z_2\in F^*_{p}$. Thus we only consider the case of $z_1=-z_2a$.
For $a\not \in F_p^*$, then we have $z_1+z_2a\not=0$ for any $z_1,z_2\in F^*_{p}.$ Thus
\begin{equation}\label{o1}
N(a,b)=\left\{\begin{array}{ll}
              p^{2m-2}-1,& a\not\in F_p^*,\\
             p^{2m-2}-1+\frac{1}{p^2}\Omega ,&a\in F_p^*.
             \end{array}\right.
\end{equation}
For $a\in F_p^*$ and $b=0$,
\begin{align}\label{o2}
  % \nonumber to remove numbering (before each equation)
    \Omega=&{q}\sum\limits_{z_2\in F^*_{p}}\sum\limits_{y\in F_{q}}\chi(-z_2ay^{p^t-1})~\ { \rm (by~Lemma~\ref{character})}\nonumber\\
     =&{q}\sum\limits_{z\in F^*_{p}}\sum\limits_{y\in F_{q}}\chi(zy^{p^t-1})={q}\sum\limits_{z\in F^*_{p}}(1+\sum\limits_{y\in F^*_{q}}\chi(zy^{p^t-1}))\nonumber\\
     =&{q(p-1)}+{q(p^t-1)}\sum\limits_{z\in F^*_{p}}       \sum\limits_{y\in \Delta}\chi(zy)\nonumber\\
     =&p^{m}(p-1)+p^{m}(p^t-1)S.
\end{align}
For $a\in F_p^*$ and $b\not=0$,
\begin{align}\label{o3}
 % \nonumber % Remove numbering (before each equation)
 \Omega
 %=&\sum\limits_{\substack{z_1,z_2\in F^*_{p},\\z_1+z_2a=0}}\sum\limits_{x\in F_{q}}\chi((z_1+z_2a)x)\sum\limits_{y\in F_{q}}\chi(z_1y^{p^t-1}+z_2by)\nonumber\\
 =&{q}\sum\limits_{z_2\in F^*_{p}}\sum\limits_{y\in F_{q}}\chi(-z_2ay^{p^t-1}+z_2by)\nonumber\\
=&{q}(p-1)+q\sum_{z_2\in F^*_{p}}\sum_{u\in F^*_{p^t}}\sum_{v\in\Gamma}\chi(-z_2av^{p^t-1}+z_2buv)\nonumber\\
=&{q}(p-1)+q\sum_{z_2\in F^*_{p}}\sum_{v\in\Gamma}\chi(-z_2av^{p^t-1})\sum_{u\in F^*_{p^t}}\chi(z_2buv)\nonumber\\
=&{q(p-1)}+{q}\sum_{z_2\in F^*_{p}}\sum_{v\in\Gamma}\chi(-z_2av^{p^t-1})(\sum_{u\in F_{p^t}}\chi(z_2buv)-1)\nonumber\\
=&{q(p-1)}+{q}\sum_{z_2\in F^*_{p}}\sum_{v\in\Gamma}\chi(-z_2av^{p^t-1})\sum_{u\in F_{p^t}}\chi(z_2buv)-{q}\sum_{z_2\in F^*_{p}}\sum_{v\in\Gamma}\chi(-z_2av^{p^t-1})\nonumber\\
=&{q(p-1)}+{q}\sum_{z_2\in F^*_{p}}\sum_{v\in\Gamma}\chi(-z_2av^{p^t-1})\sum_{u\in F_{p^t}}\chi_t(z_2Tr_t^m(bv)u)-{q}\sum_{z\in F^*_{p}}\sum_{x\in\Delta}\chi(zx)\nonumber\\
=&{q(p-1)}+{qp^t}\sum_{z_2\in F^*_{p}}\sum_{v\in\Gamma,Tr_t^m(bv)=0}\chi(-z_2av^{p^t-1})-{p^m}S~\ { \rm (by~Lemma~\ref{character})}\nonumber\\
=&{q(p-1)}+{qp^t}\sum_{z_2\in F^*_{p}}\chi(-z_2av_{b}^{p^t-1})-{p^m}S\nonumber\\
=&{q(p-1)}+{qp^t}(\sum_{z_2\in F_{p}}\chi_1(-z_2aTr_1^m(v_{b}^{p^t-1}))-1)-{p^m}S\nonumber\\
=&\left\{\begin{array}{ll}
             (p^m+p^{m+t})(p-1)-p^mS ,& Tr_1^m(v_{b}^{p^t-1})=0,\\
         p^m(p-1)-p^{m+t}-p^mS ,&Tr_1^m(v_{b}^{p^t-1})\not=0.
             \end{array}\right.
\end{align}
By Equations (\ref{o1}), (\ref{o2}) and (\ref{o3}), we can complete our proof.
%Then $$N(a,b)=\left\{\begin{array}{ll}
%  p^{2m-2}-1+(p^{m-2}+p^{m+k})(p-1)-p^{m-2}S ,& Tr_1^m(v_{b}^{p^k-1})=0;\\
%         p^{2m-2}-1+p^{m-2}(p-1)-p^{m+k-2}-p^{m-2}S ,&Tr_1^m(v_{b}^{p^k-1})\not=0.
%             \end{array}\right.$$
\end{proof}

We next discuss the dual code of the linear code $C_D$ in Theorem \ref{D}. The Lemma \ref{dual} can be used to the computation of weight distribution of $C_D$.

\begin{lemma}\label{dual}
Let $d^\bot$ denote the minimum distance of the $C^{\bot}_D,$ then $2\leq d^\bot\leq4$.
\end{lemma}

\begin{proof}
Since $(0,0)\not\in D\subseteq F_q\times F_q$, the minimum distance  $d^{\bot}$ of $C_D^{\bot}$ cannot be 1 and then $d^{\bot}\geq 2$. Since $C_D^{\bot}$ is a linear code with parameters $[p^{2m-1}-1, p^{2m-1}-1-2m]$, then we can obtain that $d^{\bot}\leq 4$ directly from Hamming bound \cite[Theorem 1.12.1]{2003-H}.
\end{proof}

\textbf{The proof of Theorem \ref{D}.}

According to Lemma \ref{length}, the length of a codeword in $C_{D}$
is $n=p^{2m-1}-1.$  Note that the weight of codeword $c(a,b)\in C_D$ is equal to $n-N(a,b)$. Then we can obtain the all weights of code $C_D$ in Table \ref{tab1} by Lemma \ref{weight}.
By Lemmas \ref{bound} and \ref{equl}, we have
\begin{equation}\label{S}
|S|\leq\sum\limits_{ z\in F_p^*}|\sum\limits_{x\in \Delta}\chi(zx)|=\sum\limits_{ z\in F_p^*}|K_t(z^2)|\leq2(p-1)\sqrt{p^t}.
\end{equation}
Thus, it is easy to prove that $w_i>0$ for any $1\leq i \leq 4$ which is equivalent to prove that the dimension of $C_D$ is $2m$.

By Lemma \ref{Pless} and $A_{1}^{\bot}=0$ by Lemma \ref{dual}, we have
$$\left\{\begin{array}{ll}
              \sum\limits_{i=1}^{4}A_{w_i}=p^{2m}-1,\\
             \sum\limits_{i=1}^{4}w_iA_{w_i}=(p-1)p^{2m-1}n.
             \end{array}\right.$$
In addition we have $A_{w_1}=p^m(p^m-p+1)-1$ and $A_{w_2}=p-1$ by Lemma \ref{weight}. Thus
we can obtain the weight distribution in Table 1 of $C_D$ by solving the system of equations.

\textbf{The proof of Corollary \ref{1D}.}

It is easy to observe that all codewords represented by $c(a,b)$ in $C_{D1}$ satisfy $a\not \in F_p^*.$  Thus we can obtain the results by Lemma \ref{weight}.

\textbf{The proof of Corollary \ref{2D}.}

From the assumption, we can know that $$C_{D2}=\{(Tr_{1}^m((\sum\limits_{i=1}^{m-1}\lambda_i\alpha^i+\nu)x+(\sum\limits_{i=1}^{m-1}\mu_i\alpha^i+\nu)y))_{(x,y)\in D}:\lambda_i, \mu_i, \nu\in F_p\}.$$

Then the number of codewords such that $\sum\limits_{i=1}^{m-1}\lambda_i\alpha^i+\nu\in F_p^*$ in $C_{D2}$ for $\lambda_i, \nu\in F_p$ is $(p-1)p^{m-1}.$
Thus $A_{w_1}=p^{2m-1}-p^{m}+p^{m-1}-1$ by Lemma \ref{weight}.

We will calculate the minimum distance $\overline{d}^\bot$ of $C_{D2}^{\bot}$. We claim that $\overline{d}^\bot\geq2.$   If $\overline{d}^\bot=1$, then there exist a nonzero element $(x_0,y_0)\in D$ such that $$Tr_{1}^{m}((\sum\limits_{i=1}^{m-1}\lambda_i\alpha^i+\nu)x_0+(\sum\limits_{i=1}^{m-1}\mu_i\alpha^i+\nu)y_0)=0$$
for any $\lambda_i, \mu_i, \nu\in F_p$.  By $(x_0,y_0)\not=(0,0)$, we assume that $x_0\not=0$ without loss of generality.  We have $\{\sum\limits_{i=1}^{m-1}\lambda_i\alpha^i+\nu:\lambda_i,\nu\in F_p\}=F_q$  and then  $\{(\sum\limits_{i=1}^{m-1}\lambda_i\alpha^i+\nu)x_0+(\sum\limits_{i=1}^{m-1}\mu_i\alpha^i+\nu)y_0:\lambda_i, \mu_i, \nu\in F_p\}=F_q$ by $x_0\not=0$.  Thus it is impossible that $\overline{d}^\bot=1$ by $\{Tr_{1}^{m}((\sum\limits_{i=1}^{m-1}\lambda_i\alpha^i+\nu)x_0+(\sum\limits_{i=1}^{m-1}\mu_i\alpha^i+\nu)y_0):\lambda_i, \mu_i, \nu\in F_p\}=F_p$, which indicates $A_{1}^{\bot}$ of $C^{\bot}_{D2}$ is equal to 0.

In addition, all codewords represented by $c(a,b)$ in $C_{D2}$ satisfy $b\not=0$ when $a\in F_p^*$.  Thus we can obtain the results in Table \ref{tab2} by Lemma \ref{Pless} and Lemma \ref{weight}.
% \nonumber % Remove numbering (before each equation)

\section{Conclusion}
 In this paper, we studied some classes of linear codes with few weights, and determined their lengths and weight distributions  according to some known results on character sums and Kloosterman sums.  From the construction of Corollaries \ref{1D} and $\ref{2D}$, we can also obtain more linear codes with the same minimum distance as $C_{D1}$ or $C_{D2}$ but lower dimension resulting in having no further discussions.

Besides, linear codes over $F_p$ have wide applications which
are used for the construction of secret sharing schemes \cite{1993-M, 2015-D}. The support of a codeword ${\bf c}=(c_0,c_1,\ldots,c_{n-1})\in C$ is defined as $supp({\bf c})=\{0\leq i\leq n-1:c_{i}\not=0\}.$
We say that a codeword $\bf{b}$ is covered by a codeword $\bf{a}$ if $supp(\bf{b})$ is the proper subset of $supp(\bf{a})$.
A nonzero codeword of a linear code $C$ is said minimal if it does not cover any other nonzero codeword of $C$.
A code $C$ is said to be minimal if every nonzero codeword is minimal.  In order to obtain secret sharing with interesting access structures, we would like to have linear codes $C$ such that
$\frac{w_{min}}{w_{max}}>\frac{p-1}{p}$ \cite{1998-A},
where $w_{win}$ and ${w_{max}}$ denote the minimum and the maximum nonzero weights in $C$, respectively.
%\begin{lemma}\label{minimal}{\rm {\cite{1998-A}}}
%Every nonzero codeword of a linear code $C$ over $F_{p}$ is minimal, provided that $$\frac{w_{min}}{w_{max}}>\frac{p-1}{p},$$
%where $w_{win}$ and ${w_{max}}$ denote the minimum and the maximum nonzero weights in $C$, respectively.
%\end{lemma}
%It is easy to check that $C_{D1}$ is minimal by Lemma \ref{minimal}.
For the linear code $C_{D1}$ of Corollary \ref{1D}, we have
$$\frac{w_{min}}{w_{max}}=1>\frac{p-1}{p},$$ and through some simple and tedious calculations, we have
$$\frac{w_{min}}{w_{max}}=\frac{w_3}{w_4}>\frac{p-1}{p}$$ by Formula (\ref{S}) for the linear code  $C_{D2}$ of Corollary \ref{2D}. Then $C_{D1}$ and $C_{D2}$ are minimal linear codes, which can be employed to obtain secret sharing schemes with interesting access structures\cite{2015-D}.

%\section*{Acknowledgments}
%This work was supported in part by the National Natural Science Foundation of China under Grant 12371326.

\end{document}